\documentclass[journal]{IEEEtran}

\usepackage{amsmath,amssymb,amsthm}
\usepackage{epsfig,graphicx}
\usepackage[svgnames,rgb]{xcolor}
\usepackage[color=yellow]{pdfcomment}
\usepackage{colortbl,hhline}
\usepackage{algorithmic,algorithm,lipsum}
\usepackage{graphicx}
\usepackage{caption}
\usepackage{subcaption}

\usepackage[mathscr]{eucal}
\usepackage{amsbsy}
\usepackage{bm}
\usepackage{fixltx2e}
\MakeRobust{\overrightarrow}
\newcommand{\Sec}[1]{\hyperref[sec:#1]{\S\ref*{sec:#1}}} 
\newcommand{\Eqn}[1]{\hyperref[eq:#1]{(\ref*{eq:#1})}} 
\newcommand{\Fig}[1]{\hyperref[fig:#1]{Figure~\ref*{fig:#1}}} 
\newcommand{\Tab}[1]{\hyperref[tab:#1]{Table~\ref*{tab:#1}}} 
\newcommand{\Thm}[1]{\hyperref[thm:#1]{Theorem~\ref*{thm:#1}}} 
\newcommand{\Lem}[1]{\hyperref[lem:#1]{Lemma~\ref*{lem:#1}}} 
\newcommand{\Prop}[1]{\hyperref[prop:#1]{Property~\ref*{prop:#1}}} 
\newcommand{\Cor}[1]{\hyperref[cor:#1]{Corollary~\ref*{cor:#1}}} 
\newcommand{\Def}[1]{\hyperref[def:#1]{Definition~\ref*{def:#1}}} 
\newcommand{\Alg}[1]{\hyperref[alg:#1]{Algorithm~\ref*{alg:#1}}} 
\newcommand{\Ex}[1]{\hyperref[ex:#1]{Example~\ref*{ex:#1}}} 

\newcommand{\V}[1]{{\bm{\mathbf{\MakeLowercase{#1}}}}} 
\newcommand{\mb}[1]{\mathbb{#1}}
\newcommand{\mc}[1]{\mathcal{#1}}

\newcommand{\M}[1]{{\bm{\mathbf{\MakeUppercase{#1}}}}} 

\title{On Deterministic Conditions for Subspace Clustering under Missing Data}
\author{Wenqi Wang, Shuchin Aeron and Vaneet Aggarwal\thanks{W. Wang and V. Aggarwal are with the School of Industrial Engineering, Purdue University, West Lafayette, IN, 47907, email: \{wang2041,vaneet\}@purdue.edu. S. Aeron is with the Dept. of Electrical and Computer Engineering, Tufts University, Medford, MA 02155, email: shuchin@ece.tufts.edu

This work was presented in part at the IEEE International Symposium on Information Theory, Spain in July 2016.

Shuchin Aeron acknowledges the support from NSF grant CCF:1319653. 
 }}

\newtheorem{definition}{Definition}
\newtheorem{lemma}{Lemma}

\newtheorem{corollary}{Corollary}
\newtheorem{theorem}{Theorem}

\newcommand{\oil}{\Omega_{i}^{(\ell)}}
\newcommand{\oi}{\Omega_i}

\newcommand{\uell}{\M{U}^{(\ell)}}

\newcommand{\vell}{\M{V}^{(\ell)}}
\newcommand{\xell}{\M{X}^{(\ell)}}
\newcommand{\Aell}{\M{A}^{(\ell)}}

\begin{document}
\pagenumbering{gobble}

\maketitle

\begin{abstract}
In this paper we present deterministic conditions for success of sparse subspace clustering (SSC) under missing data, when data is assumed to come from a Union of Subspaces (UoS) model. We consider two algorithms, which are variants of SSC with entry-wise zero-filling that differ in terms of the optimization problems used to find affinity matrix for spectral clustering. For both the algorithms, we provide deterministic conditions for any pattern of missing data such that perfect clustering can be achieved.
We provide extensive sets of simulation results for clustering as well as completion of data at missing entries, under the UoS model. Our experimental results indicate that in contrast to the full data case, accurate clustering does not imply accurate subspace identification and completion, indicating the natural order of relative hardness of these problems.
\end{abstract}

\section{Introduction}

In this paper we consider the problem of data clustering under the union of subspaces (UOS) model \cite{soltanolkotabi2012,Elhamifar:2012uz}, also referred to as a subspace arrangement \cite{Ma_SIAM2008, TsakirisV15}, when each data vector is sub-sampled. This is referred to as the case of \emph{missing data}. In other words we are looking to harvest a union of subspaces structure from the data, when the data is missing. Such a problem has been recently considered in a number of papers \cite{bal2,bal3,ErikssonArXiv2011,Yang2015,pimentelgroup}. This setting has implications to data completion under the union of subspaces model in contrast to the single subspace model that has been prevalent in the matrix completion literature. In contrast to statistical analysis in \cite{bal2,bal3,ErikssonArXiv2011}, this paper uses a variant of the sparse subspace clustering (SSC) algorithm \cite{Elhamifar:2012uz} to give sufficient deterministic conditions for accurate subspace clustering under missing data.

We consider two algorithms in this paper. Both algorithms are based on  Sparse Subspace Clustering (SSC) with entry-wise zero-filling (EWZF). The first algorithm (SSC-EWZF) represents each entry-wise zero-filled data point as a sparse linear combination of other entry-wise zero-filled data points. These coefficients are then used to do spectral clustering \cite{Ng01onspectral} which gives the desired clusters. The second algorithm (SSC-EWZF-OO) represents observed entries of each data point as a sparse linear combination of other entry-wise zero-filled data points when projected onto the observed entries of the represented data point. Thus, the overlapping observations (OO) are used to determine the sparse linear combination. The first algorithm tries to match the zero filled data point when representing with the non-zero points in the other vectors which might result in inaccuracies. On the other hand the second algorithm matches the data only at its sampled locations and is therefore more robust.
In this paper, we derive deterministic conditions for any sampling pattern for these algorithms to succeed in producing the correct clustering.  We note that the conditions for correct clustering are simple and readily interpretable for the case when all the data points are sampled at exactly the same locations and directly reduce to the  results in \cite{dim_red_sc_ext} when data is fully sampled. 



We numerically compare the performance of the proposed algorithms with two other algorithms, SSC-EWZF-OO-LASSO \cite{Yang2015}, which is a LASSO version of SSC-EWZF-OO, and zero-filled version of Thresholded Subspace Clustering (TSC) \cite{tsc2014} denoted by TSC-EWZF. 
Our proposed algorithm SSC-EWZF-OO is comparable with SSC-EWZF-OO-LASSO algorithm and shows the best accuracy in the numerical results. 
In contrast, SSC-EWZF algorithm and TSC-EWZF algorithm perform worse than the other two algorithms.
Further we consider the error in subspace recovery after clustering the data points. To the best of our knowledge this is the first time it is demonstrated that accurate clustering under missing data does not imply accurate \emph{subspace identification} and \emph{data completion} thereby indicating the natural order of hardness of these problems under missing data.

The rest of the paper is organized as follows. Section II describes the problem of subspace clustering under missing data. Two algorithms considered for solving the problem are outlined in Section III. Section IV analyzes the deterministic conditions for any given sampling patterns when each of the algorithm gives correct clustering results. These deterministic conditions are specialized to the case of  when all the data points are sampled at the same locations and no missing data. Numerical results are provided in Section V, and Section VI concludes the paper. All proofs are provided in the Appendix.

\section{Problem set-up}
\label{sec:prob_setup}

We are given a set of data points collected as columns of a matrix $\M{X} \in \mathbb{R}^{n\times N}$, i.e. $\M{X}_{i} \in \mathbb{R}^n, i = 1,2,...,N$ from union of $L$ subspaces such that, $\M{X}_i \in \bigcup_{\ell =1}^{L} \mb{S}^{(\ell)}$, where $\mb{S}^{(\ell)}$ is a subspace of dimension $d_\ell$ in $\mb{R}^{n}$, for $\ell = 1,2,...,L$. 
We let $d$  be the maximum of the dimension of $L$ subspaces, or $d \triangleq \max_{\ell} d_\ell$.
Each data point $\M{X}_i$ is sampled at $\Omega_i$ co-ordinates (randomly or deterministically), denoted as ${\bf X}_{\Omega_i}$.
In order to derive meaningful performance guarantees for the proposed algorithm, we consider the following generative model for the data. 
Let $\xell \in \mathbb{R}^{n\times N_\ell}$ denote the set of vectors in $\M{X}$ which belong to subspace $\ell$. 
Let $$ \xell = \uell \Aell ,$$ where the $N_\ell$ columns of $\Aell \in \mathbb{R}^{d_\ell \times N_\ell}$ are drawn from the unit sphere $\mc{S}^{d_\ell-1}$ and $\uell \in \mathbb{R}^{n\times d_\ell}$ is a matrix with orthonormal columns, whose columns span the subspace, $\mb{S}^{(\ell)}$.  
Let ${\bf a}_i^{(\ell)} \in \mathbb{R}^{d_\ell \times 1}$ be the $i_\text{th}$ column of ${\bf A}^{(\ell)}$. Then under missing data, point $\xell_{\oi}$ is the $i_\text{th}$ data from $\mathbb{S}^{(\ell)}$ sampled at locations $\oil$, denoted as
\begin{align} 
\xell_{\oi} = \mathbf{I}_{\oil} \uell \V{a}_{i}^{(\ell)}
\end{align}
where $\mathbf{I}_{\oil}$ is a diagonal matrix with $\mathbf{I}_{\oil}(k,k) = 1$ iff $ k \in \oil$. It is essentially a zero filled $\M{X}_{i}^{(\ell)}$ at missing entires.  \\

Given this set-up the problem is to accurately cluster the data points such that within each cluster the data points belong to the same (original) subspace. \\

In the following, we will use lowercase boldface letters to represent column vectors and uppercase boldface to designate matrices. The superscript $^\top$ denotes conjugate transpose and $^\dagger$ denotes the Moore-Penrose pseudo-inverse. If ${\bf V}$ has full column rank, then ${\bf V}^\dagger$ is a left inverse of ${\bf V}$, expressed as ${\bf V}^\dagger = ({\bf V}^\top {\bf V})^{-1} {\bf V}^\top$ and if
${\bf V}$ has full row rank, then ${\bf V}^\dagger$ is a right inverse of ${\bf V}$, expressed as ${\bf V}^\dagger = {\bf V}^\top ({\bf V}^\top {\bf V})^{-1} $. 

In addition, the following notations  will be extensively used in our analysis. 
\begin{enumerate}
\setlength \itemsep{3pt}
\item Let $\M{V}_{\oi}^{(\ell)} \triangleq \M{I}_{\oil} \M{U}^{(\ell)}$ denote the  truncated basis of ${\bf X}_{\Omega_i}^{(\ell)}$. The singular value decomposition (SVD) of $\M{V}_{\oi}^{(\ell)}$ is given as  $\M{V}_{\oi}^{(\ell)} = \M{Q}_{i}^{(\ell)} {\bf \Sigma}_{i}^{(\ell)} \M{R}_{i}^{(\ell) \top}$, where 
${\bf Q}_i^{(\ell)} \in \mathbb{R}^{n\times n}, {\bf \Sigma}_i^{(\ell)} \in \mathbb{R}^{n \times d}, {\bf R}_i^{(\ell)} \in \mathbb{R}^{d\times d}$. 
Thus,
\begin{equation} \label{eq: X_omega}
{\bf X}_{\Omega_i}^{(\ell)} = \M{Q}_{i}^{(\ell)} \Sigma_{i}^{(\ell)} \M{R}_{i}^{(\ell) \top} {\bf a}_i^{(\ell)}.
\end{equation}
\item Let 
\begin{equation} \label{eq: base}
\tilde{{\bf a}}_j^{(k)} \triangleq  ({\bf Q}_i^{(\ell)})^\top {\bf I}_{\Omega_i^{(\ell)}} {\bf I}_{\Omega_j^{(k)}} {\bf U}^{(k)}  {\bf a}_j^{(k)}.
\end{equation}

In a special case  when $j=i$ and $k=l$, we have
\begin{equation} \label{eq: notation_a}
\begin{split}
\tilde{{\bf a}}_i^{(\ell)}  
& = ({\bf Q}_i^{(\ell)})^\top {\bf I}_{\Omega_i^{(\ell)}} {\bf I}_{\Omega_i^{(\ell)}} {\bf U}^{(\ell)}  {\bf a}_i^{(\ell)} \\
& = ({\bf Q}_i^{(\ell)})^\top {\bf I}_{\Omega_i^{(\ell)}} {\bf U}^{(\ell)}  {\bf a}_i^{(\ell)} \\
& = ({\bf Q}_i^{(\ell)})^\top {\bf Q}_i^{(\ell)} {\bf \Sigma}_i^{(\ell)}({\bf R}_i^{(\ell)})^\top {\bf a}_i^{(\ell)} \\
& = {\bf \Sigma}_i^{(\ell)}({\bf R}_i^{(\ell)})^\top {\bf a}_i^{(\ell)}.
\end{split}
\end{equation}

Further, let
\begin{equation} \label{eq: base_groups}
\tilde{\M{A}}_{-i}^{(\ell)} \triangleq [\tilde{{\bf a}}_1^{(\ell)} , \tilde{{\bf a}}_2^{(\ell)} ,..., \tilde{{\bf a}}_{i-1}^{(\ell)} , \tilde{{\bf a}}_{i+1}^{(\ell)} ,..., \tilde{{\bf a}}_{N_l}^{(\ell)} ],
\end{equation}
be  $n \times (N_\ell-1)$ matrices with columns as $\tilde{\V{a}}_{j}^{(\ell)},  j \neq i$ and 
\begin{equation}
{\M{A}}_{-i}^{(\ell)} \triangleq [{{\bf a}}_1^{(\ell)} , {{\bf a}}_2^{(\ell)} ,..., {{\bf a}}_{i-1}^{(\ell)} , {{\bf a}}_{i+1}^{(\ell)} ,..., {{\bf a}}_{N_l}^{(\ell)} ],
\end{equation}
be $d \times  (N_\ell-1)$ matrices with columns as ${\V{a}}_{j}^{(\ell)}, j \neq i$.
\end{enumerate}

We now introduce several geometric definitions that are used to state the main results. 


\begin{definition}[Centro-Symmetric Polytope] 
For any matrix ${\bf P} \in \mathbb{R}^{n \times N}$, $\mc{P}({\bf P})$ denotes the centro-symmetric polytope defined as, $\mc{P}({\bf P}) = \text {conv} (\pm {\bf p}_1, \pm {\bf p}_2,..., \pm {\bf p}_N)$, where $\text{conv}(\cdot)$ denotes the convex hull operation of the points in the argument.
\end{definition}

\begin{definition}[Inradius \cite{gritzmann1992inner}]The in-radius of $\mc{P}$, denoted as $r(\mc{P})$, is defined as the radius of the largest Euclidean ball that can be inscribed in $\mc{P}$.
	\end{definition}
	
\begin{definition}[Circumradius \cite{gritzmann1992inner}] 
The circumradius of $\mc{P}$, denoted as $R(\mc{P})$, is defined as the radius of the smallest Euclidean ball that contains $\mc{P}$.
\end{definition}

\begin{definition}[Polar Set \cite{boyd2004convex}]
The polar set of $\mc{P}({\bf P})$ is given by,
\begin{equation}
\mc{P}^o({\bf P}) = \{ {\bf z:} \|{\bf P} ^\top z\|_\infty \leq 1\}.
\end{equation}
\end{definition}

\section{Algorithm}

We present two algorithms based on entry-wise zero-filling variants of the Sparse Subspace Clustering  (SSC) algorithm \cite{vidal}. These algorithms essentially form an affinity matrix between data points based on finding sparse self-representation of the data. This affinity matrix is then subsequently used for Spectral Clustering \cite{Ng01onspectral} to find the clusters.


The first algorithm, denoted SSC-EWZF, zero-fills the missing entries and applies the SSC algorithm. The detailed steps are given as follows.  

\noindent {\bf Algorithm SSC-EWZF:}
\begin{enumerate}
\item For each $i$ solve for 
\begin{equation} \label{eq: algo}
\arg \min \|\V{c}_{i}\|_1: \,\, \M{X}_{\oi} =  \M{X}_{-i,\Omega} \V{c}_{i},
\end{equation}
where $\M{X}_{\oi}$ denotes the data point $\M{X}_i$ with zeros filling at non-sampled locations and $\M{X}_{-i,\Omega}$ denotes the zero-filled data points except the $i$ data point. 
\item Collect the $\V{c}_i$ into a matrix $\M{C}$ and apply spectral clustering to $\M{A} = |\M{C}| + |\M{C}|^\top$
\end{enumerate}

The second algorithm that we consider, denoted  SSC-EWZF-OO,  represents observed entries of each data point as a sparse linear combination of other entry-wise zero-filled data points when projected onto the observed entries of the represented data point. The detailed steps are given as follows.

\noindent {\bf Algorithm SSC-EWZF-OO:}
\begin{enumerate}
\item For each $i$ solve for 
\begin{equation} \label{eq: algo}
\arg \min \|\V{c}_{i}\|_1: \,\, \M{X}_{\oi} = \mathbf{I}_{\oi} \M{X}_{-i,\Omega} \V{c}_{i}
\end{equation}
where $\M{X}_{\oi}$ denotes the data point $\M{X}_i$ with zeros filling at non-sampled locations and $\M{X}_{-i,\Omega}$ denotes the zero-filled data points except the $i$ data point. 
\item Collect the $\V{c}_i$ into a matrix $\M{C}$ and apply spectral clustering  to $\M{A} = |\M{C}| + |\M{C}|^\top$
\end{enumerate}

In the following we will analyze these two algorithms to derive conditions under which, the points are correctly clustered.

	
	

\section{Analysis of the Algorithms}
We begin by noting the following result for SSC-EWZF-OO.



\subsection{Deterministic Conditions for SSC-EWZF-OO}

The following theorem provides the deterministic conditions for subspace clustering under missing data when SSC-EWZF-OO is used. 
{\theorem \label{theorem}
Let $|\Omega_i^{(\ell)}| \geq d$. Then SSC-EWZF-OO leads to correct clustering if for all $i\in [N_\ell], k\neq \ell$, the following holds
\begin{equation}\label{eq:theorem_eq}
\left| 
\frac{({\V \lambda}_i^{(\ell)})^\top}{\|({\V \lambda}_i^{(\ell)})^\top\|_2} ({\bf Q}_i^{(\ell)})^\top {\bf I}_{\Omega_i^{(\ell)}} {\bf I}_{\Omega_j^{(k)}} {\bf U}^{(k)} {\bf a}_j^{(k)}
\right|
< 
r(\mc{P}(\tilde{{\bf A}}_{-i}^{(\ell)}))
\end{equation}
where\\
(1) ${\V \lambda}_i^{(\ell)} \in \arg \max_{\V \lambda} \langle \tilde{{\bf a}}_i^{(\ell)}, {\V \lambda} \rangle \quad s.t. \quad \|(\tilde{{\bf A}}_{-i}^{(\ell)})^\top {\V \lambda}\|_\infty \leq 1$\\
(2) $r(\mc{P}(\tilde{{\bf A}}_{-i}^{(\ell)}))$ is the in-radius of $\mc{P}(\tilde{{\bf A}}_{-i}^{(\ell)})$. 
\endtheorem
}
\begin{proof}
	The proof is provided in Appendix \ref{apdx:oo}.
\end{proof}

The in-radius { $r(\mc{P}(\tilde{\M{A}}_{-i}^{(\ell)}))$}  depends on the sampling patterns within the subspace and is also dependent on the particular sampling pattern for the $i$-th data point. 



\begin{figure}[htbp]
\centering \makebox[0in]{
    \begin{tabular}{c c}
      \includegraphics[scale=0.45]{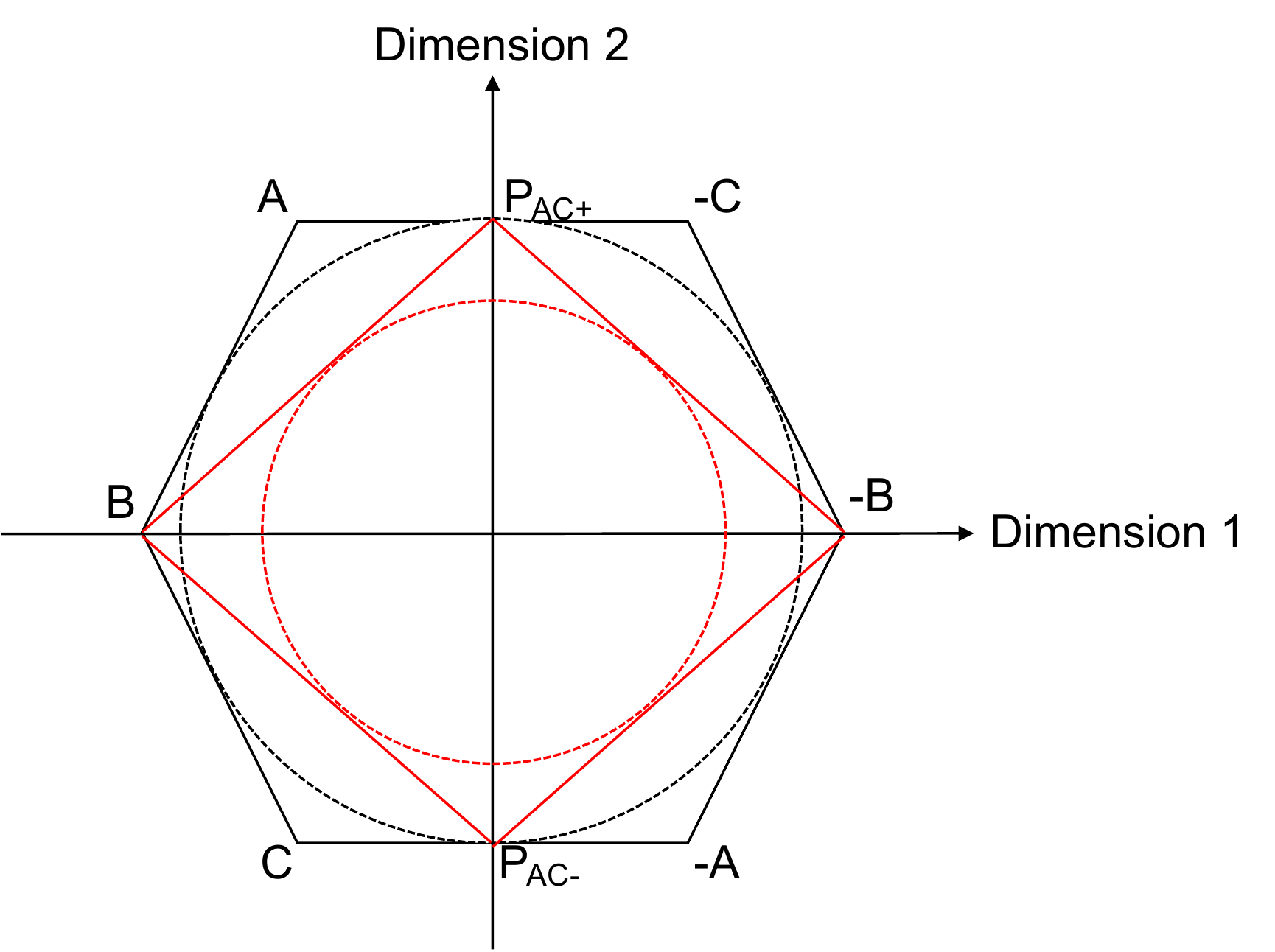}
 \end{tabular}}
  \caption{ Change of in-radius in 2-Dimensional Case}
  \label{fig:InRadius}
\end{figure}

Note that missing data will only decrease the in-radius as the missing data draws points towards the missing dimensions. 
For example, as shown in Fig \ref{fig:InRadius} where the centro-symmetric polytope under full observation is constructed by $ABC$ and the in-radius is radius of the black-dashed circle. Zero-filling the missing entires for $A$ and $C$ in Dimension 2 projects $A$, $-C$ to $P_{AC+}$ and $C$, $-A$ to $P_{AC-}$, resulting the reduced centro-symmetric body  with a smaller in-radius, which is the radius of the red-dashed circle. 
In general, the in-radius of the centro-symmetric body is large if data sample size is large and if each data is \emph{well-distributed} on the unit sphere. However, missing data projects points onto low dimensional space which degrades the uniformness of the data distribution, thus resulting in the decrease of the in-radius. 
%



\subsection{Deterministic Conditions for SSC-EWZF}
The following theorem give the deterministic conditions for subspace clustering under missing data when  SSC-EWZF is used. 

\begin{theorem} \label{corollary3}
Let $|\Omega_i^{(\ell)}| \geq d$.  Then SSC-EWZF leads to correct clustering if for all $i \in [N_\ell], k\neq \ell$, the following holds
\begin{equation}
\label{eq:corr3}
\left| 
{ \frac{({\V \lambda}_i^{(\ell)})^\top}{\|({\V \lambda}_i^{(\ell)})^\top\|_2} ({\bf Q}_i^{(\ell)})^\top  }
{  {\bf V}_{\Omega_j}^{(k)} {\bf a}_j^{(k)}}
\right|
< 
r(\mc{P}(\tilde{{\bf B}}_{-i}^{(\ell)}))
\end{equation}
where\\
(1) $\tilde{\bf b}_j^{(k)} = ({\bf Q}_i^{(\ell)})^\top {\bf V}_{\Omega_j}^{(k)} {\bf a}_j^{(k)}$, $\tilde{\bf B}_{-i}^{(\ell)}$ is the $n\times (N_\ell-1)$ matrix with columns as $\tilde{{\bf b}}_j^{(\ell)}, j\neq i$. \\
(2) ${\V \lambda}_i^{(\ell)} \in \arg \max_{\V \lambda} \langle \tilde{{\bf b}}_i^{(\ell)}, {\V \lambda} \rangle \quad s.t. \quad \|(\tilde{{\bf B}}_{-i}^{(\ell)})^\top {\V \lambda}\|_\infty \leq 1$\\
\end{theorem}
\begin{proof}
	The proof is provided in Appendix \ref{apdx:res}.
\end{proof}
Note that in Theorem \ref{theorem} the in-radius term on the RHS increases compared to the in-radius term on the RHS of \eqref{eq:corr3} in Theorem \ref{corollary3}, 
while the incoherence term on the LHS of \eqref{eq:theorem_eq} is undetermined compared with that on the LHS of \eqref{eq:corr3} as the projection ${\bf I}_{\Omega_i^{\ell}}$ may both increase or decrease the values.
Therefore, based on these two results we \emph{cannot} conclude whether SSC-EWZF-OO is better compared to SSC-EWZF or vice versa. 
It may be the case that on an average
the incoherence term on the LHS increases smaller than the increase of the in-radius term on the RHS in \eqref{eq:theorem_eq}, thus SSC-EWZF-OO shows better performance than SSC-EWZF  as shown in Section \ref{sec:sims}. 
This remains an important avenue for future research.

\subsection{Deterministic Conditions in Special Scenarios}
In this section, we will consider two special cases for the two algorithms. 

\subsubsection{Case 1}
The first case is when all the data points are sampled at the same locations, or $\Omega_{i} = \Omega$  for all $i$.  
In this case, the basis of subspace is only subspace dependent, not data dependent. Thus let ${\bf V}_\Omega^{(\ell)}$ denote the basis for subspace $\mathbb{S}^{(\ell)}$ when data are sampled at the same locations. 
We  note that both the Algorithms SSC-EWZF and SSC-EWZF-OO are the same in this special case. 


\begin{theorem}
\label{thm:1}
Let $\oil = \Omega$ for all $i, \ell$ and $|\Omega| \geq d$.  Then SSC-LP leads to correct clustering if for all $ i \in [N_\ell]$, $k \neq \ell$, the following holds,
\begin{align}
\label{eq:case1}
\left| \frac{ \V{{\bar \lambda}}_{i}^{(\ell) \top}}{\| \V{\bar{\lambda}}_i^{(\ell)}\|_2} (\vell_{\Omega})^{\dagger} \M{V}_{\Omega}^{(k)} \V{a}_{j}^{(k)} \right| < r(\mc{P}({\M{A}}_{-i}^{(\ell)}))\,\, ,
\end{align} 
where $\V{\bar{\lambda}}_i^{(\ell)} \in \arg \max_{\V{\lambda}} \langle \V{a}_{i}^{(\ell)}, \V{{ \lambda}} \rangle\,\, : \,\, \| \M{A}_{-i}^{(\ell) \top} \V{\lambda}\|_{\infty} \leq 1$.
\end{theorem}
\begin{proof}
	The proof is provided in Appendix \ref{apdx:same2}.
\end{proof}

\subsubsection{Case 2}
We now consider the case when none of the data is missing. In this special case, both SSC-EWZF and SSC-EWZF-OO algorithms are the same too. We note that the deterministic conditions in this case reduce to the same as in  \cite{dim_red_sc_ext}. The deterministic conditions in Theorem \ref{theorem} can be specialized to this case as follows. 

\begin{corollary}\label{corollary1}  
In the case that data is fully observed, the deterministic conditions in Theorem \ref{corollary3} and Theorem \ref{theorem} converts as follows. SSC leads to correct clustering if for all $ i \in [N_\ell]$, $k \neq \ell$, the following holds,
\begin{equation}\label{eq:case 1}
\left| 
\frac{(\bar{{\V \lambda}}_i^{(\ell)})^\top} {\|\bar{{\V \lambda}}_i^{(\ell)}\|_2} 
({\bf U}^{(\ell)})^\dagger
{\bf U}^{(k)} {\bf a}_j^{(k)}
\right|
< r(\mc{P}({{\bf A}}_{-i}^{(\ell)}))
\end{equation}
where 
$ \bar{{\V \lambda}}_i^{(\ell)} \in \arg \max_{\V \lambda} \langle {{\bf a}}_i^{(\ell)}, {\V \lambda} \rangle \quad s.t. \quad \|({{\bf A}}_{-i}^{(\ell)})^\top {\V \lambda}\|_\infty \leq  1 $. 
This is the same result as in Theorem 7 in \cite{dim_red_sc_ext}.
\end{corollary}

We note that Corollary \ref{corollary1} is a special case of Theorem \ref{thm:1} in which ${\bf I}_\Omega^{(\ell)}= {\bf I}_\Omega^{(k)} = {\bf I}$ such that ${\bf V}_\Omega^{(\ell)} = {\bf U}^{(\ell)}$, ${\bf V}_\Omega^{(k)} = {\bf U}^{(k)}$.
Since the main difference between \eqref{eq:case 1} and \eqref{eq:case1} comes from the change of basis from ${\bf  U}^{(\ell)}$ to ${\bf V}_\Omega^{(\ell)}$. Therefore, the performance degradation in clustering comes from increase in the left hand side (LHS) in \eqref{eq:case1} on an average under missing data.
The increase in LHS of \eqref{eq:case1}  depends on how badly conditioned $\M{V}_{\Omega}^{(\ell)}$ is. If some of the singular values of $\M{V}_{\Omega}^{(\ell)}$ are very small, it will make the LHS in Equation \eqref{eq:case1} large. To further understand this, let us consider a semi-random model where in each subspace the data points are generated by choosing $\V{a}_{i}^{(\ell)}$ uniformly randomly on a unit sphere \cite{soltanolkotabi2012}. In this case it is easy to show that $\dfrac{\bar{\V{\lambda}}_i^{(\ell)} }{\|\bar{\V{\lambda}}_i^{(\ell)}\|_2}$ are also uniformly distributed on the unit sphere \cite{dim_red_sc_ext} and the {expected value} of the LHS becomes $ \dfrac{\|(\vell_{\Omega})^{\dagger} \M{V}_{\Omega}^{(k)}\|_{F}}{d}$. This is essentially the (unnormalized) co-ordinate restricted coherence between subspaces $\ell$ and $k$.

If the subspaces are sufficiently incoherent with the standard basis of satisfy an RIP like property for large $|\Omega|$, then the condition number of ${\bf V}_\Omega^{(\ell)}$ is controlled and one can expect to obtain similar performance as the full observation case.

Note that while in this setting one can ensure perfect clustering, one cannot ensure either perfect subspace identification or completion. This is because in this case the deterministic necessary conditions for identification and completion \cite{2014arXiv1410.0633P} are not satisfied. {This indicates that clustering is an easier problem compared to subspace identification and completion under missing data.}



 \begin{figure}[htbp]
\centering
\vspace{-25mm}
\includegraphics[ width = .45\textwidth]{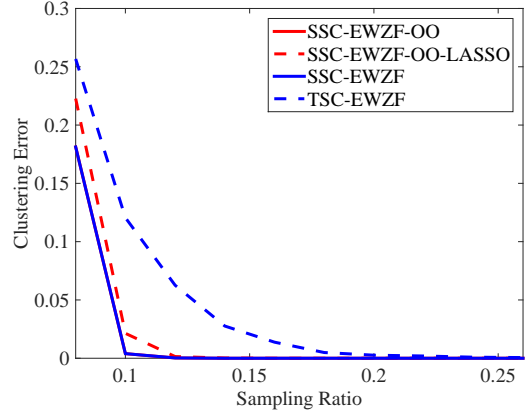} 
 \vspace{-20mm}
\caption{Clustering error in special Case 1 (points sampled at same $pn$ co-ordinates)  for varying $p$.}
\label{simulation0}
\end{figure}

\begin{figure*}[t!]
	\centering
	\begin{subfigure}[b]{0.45\textwidth}
		\centering
	\includegraphics[trim=.45in 1.5in .8in 2.8in, clip, width = \textwidth]{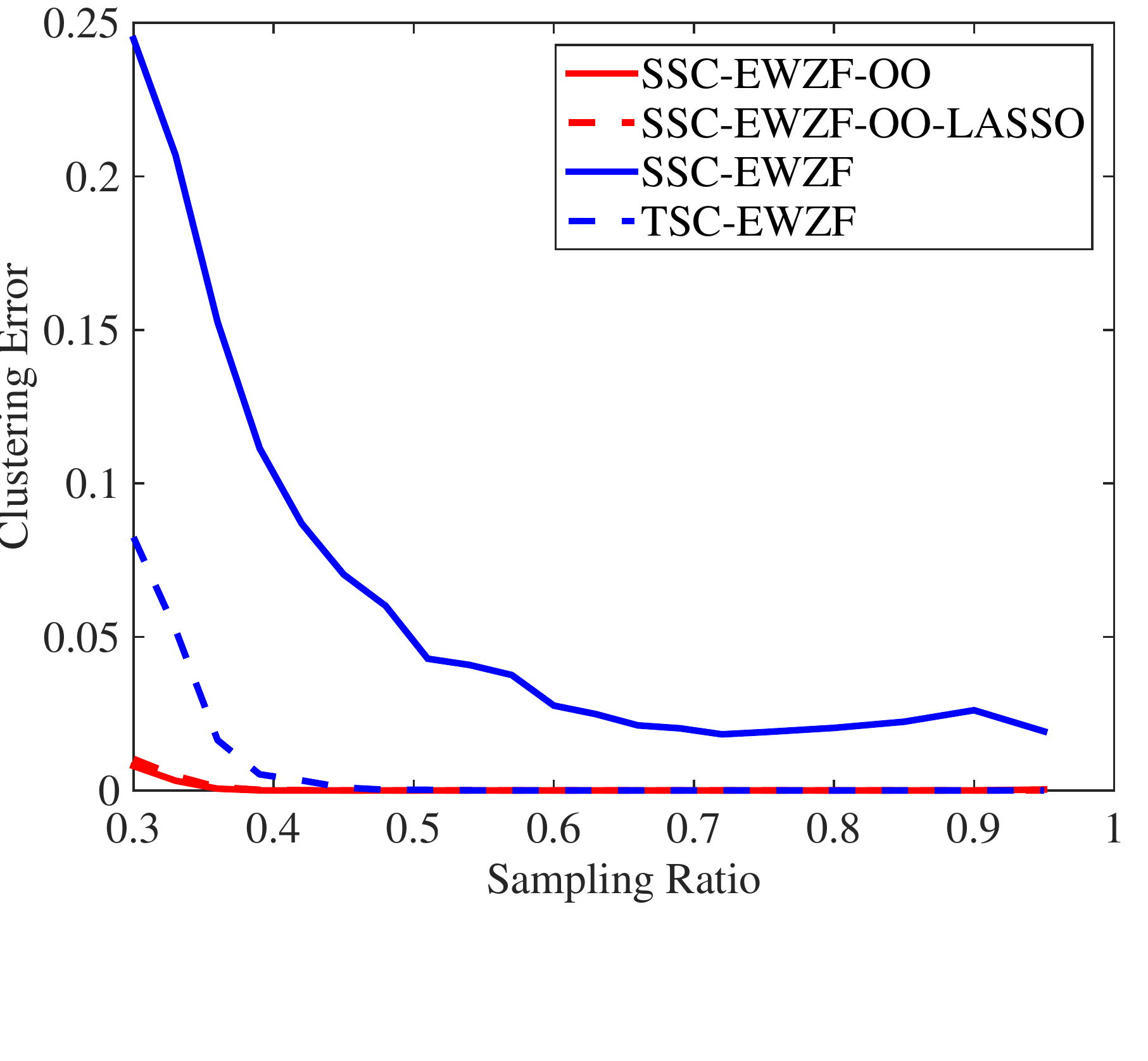}\vspace{-.3in}
		\caption{Clustering Error}
	\end{subfigure}%
	~ 
	\begin{subfigure}[b]{0.45\textwidth}
		\centering
	\includegraphics[trim=.6in 1.5in .8in 2.8in, clip, width = \textwidth]{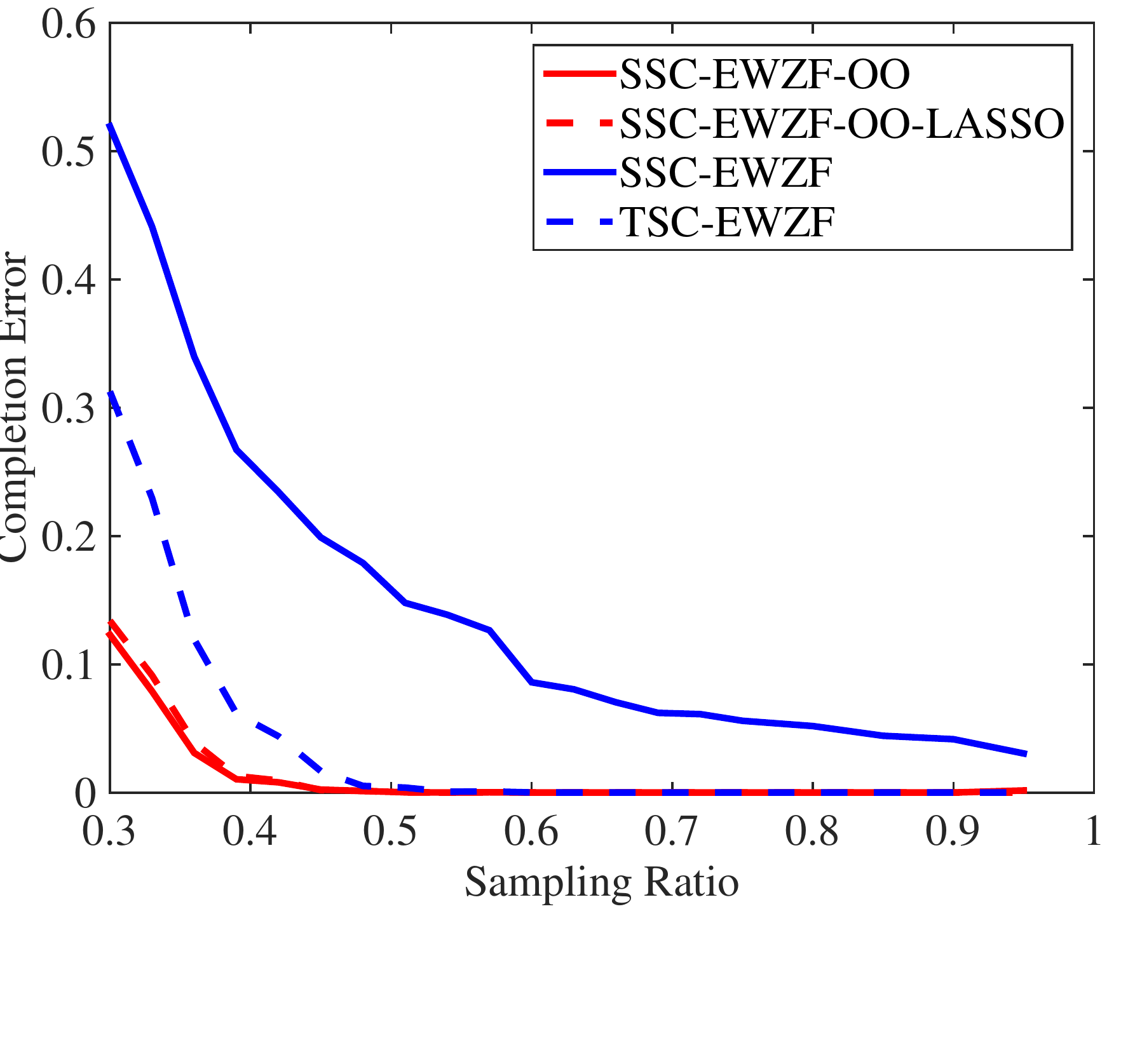}\vspace{-.3in}
		\caption{Completion Error}
	\end{subfigure}
	~ 
		\begin{subfigure}[b]{0.45\textwidth}
			\centering
			      \includegraphics[trim=.5in 1.5in .8in 2.8in, clip, width = \textwidth]{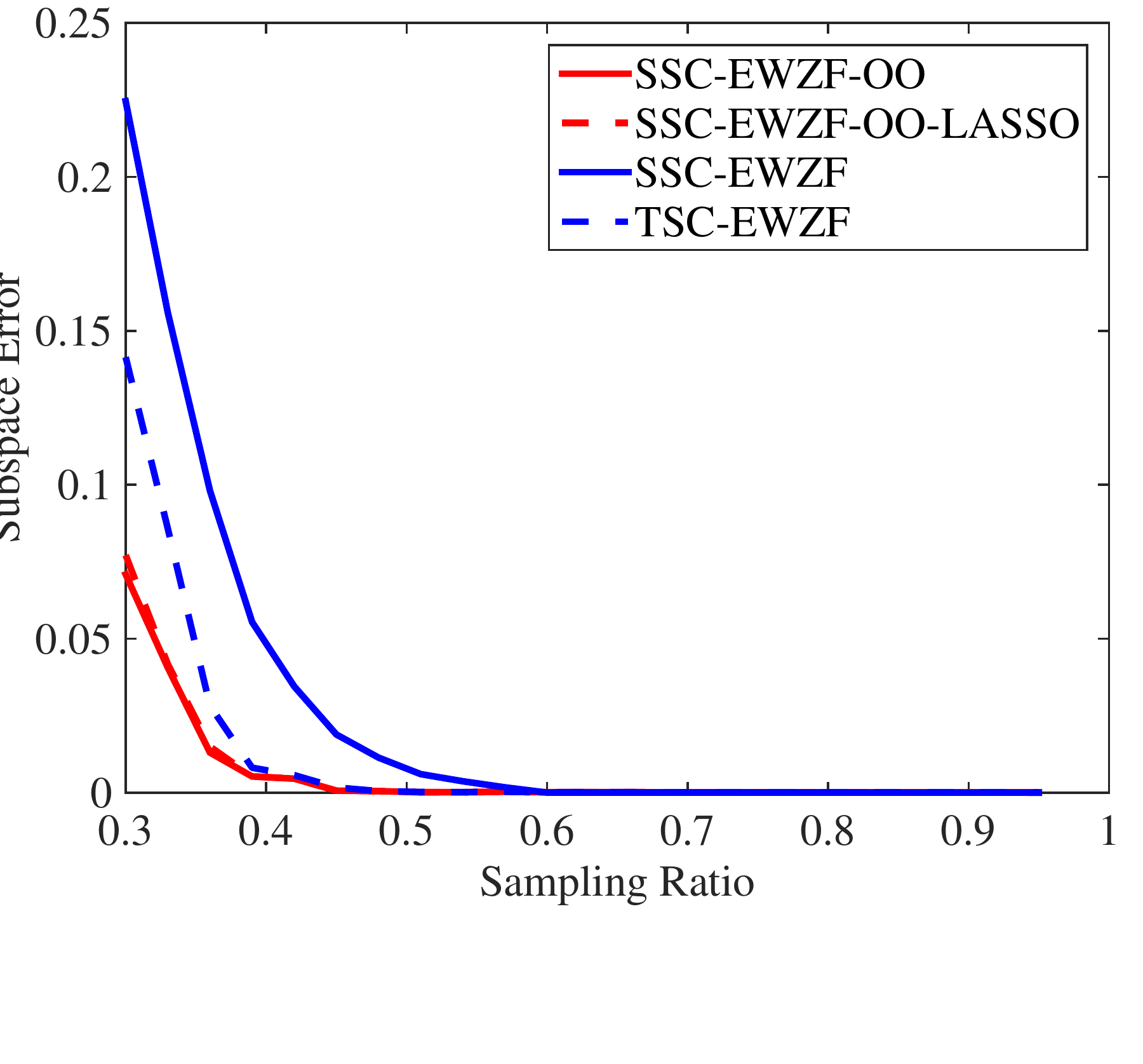}\vspace{-.3in}
			      \caption{Orthonormal Basis based Subspace Error}
		\end{subfigure}
	~
			\begin{subfigure}[b]{0.45\textwidth}
			\centering
			      \includegraphics[trim=.5in 1.5in .8in 2.8in, clip, width = \textwidth]{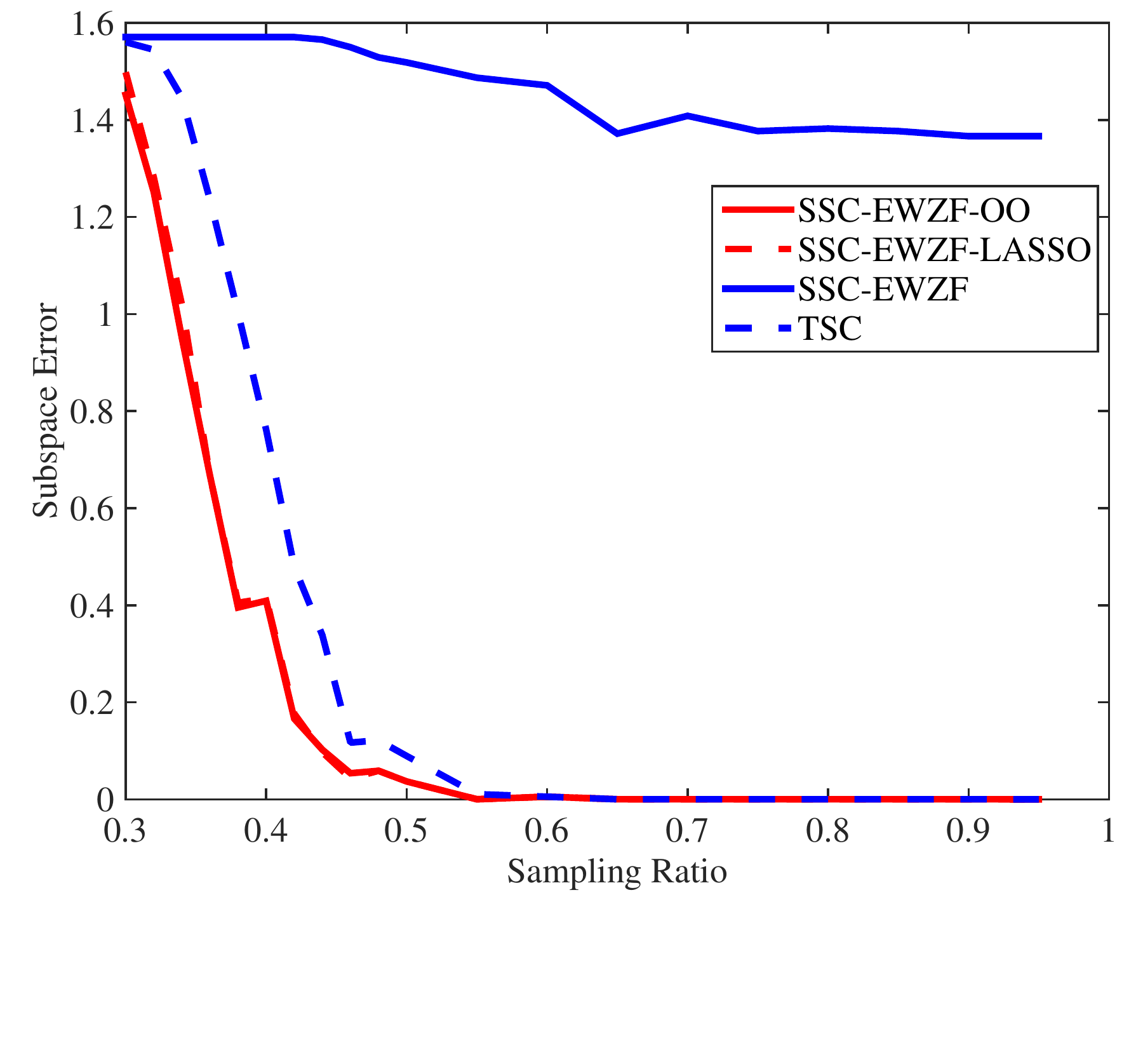}\vspace{-.3in}
			      \caption{Grassmann Metric based Subspace Error}
		\end{subfigure}
	\caption{Different error metrics  for the case of random sampling.} \label{fig:tensor}\vspace{-.1in}
	\label{simulation}
\end{figure*}



\section{Numerical Results}
\label{sec:sims}

In this section, we will see the numerical performance of the proposed algorithms, SSC-EWZF-OO and SSC-EWZF, as compared to two other baseline algorithms for data clustering and completion under the UoS model with missing data. The data is generated by using $L$ rank $d$ subspaces, each composed of $N_\ell$ vectors of dimension $n$. We assume $n=50,\  L=3,\   d=3,$ and $N_\ell=150$ for the numerical results such that there are enough data to identify the subspace \cite{pimentel2015characterization}. 
Data in each subspace is generated by a multiplication of standard entry-wise Gaussian distributed $n\times d$ matrix and a standard entry-wise Gaussian distributed $d\times N_\ell$ matrix. 

We compare our algorithms with two other algorithms. The first is  SSC-EWZF-OO-LASSO \cite{Yang2015}, which solves the LASSO problem rather than the $\ell_1$ norm minimization problem under linear equality constraint in this paper. This algorithm is selected as it gives the lowest clustering error among the different algorithms considered in \cite{Yang2015} (e.g., SSC-EC, SSC-CEC, ZF-SSC, MC-SSC).  The performance of  SSC-EWZF-OO-LASSO algorithm largely depends on the choice of $\lambda$, which is chosen to be
\begin{align}
\lambda =\frac{\alpha}{\max_{i\neq j} |X_{\Omega_i}^\top X_{\Omega _j}|_{ij}},
\end{align}
where $\alpha $ is a tuning parameter\cite{Yang2015},  set to be 7.34 for our experiment (  selected by performing an optimized performance for different values of $\alpha$). 
The second is TSC-EWZF algorithm proposed in \cite{tsc2014}, which builds adjacency matrix by thresholding correlations.  The thresholding parameter $q$ is given as
\begin{align}
q=\sqrt{N_\ell\log (N_\ell) },
\end{align}
such that $q$ is of an order smaller than  $N_\ell$ and larger than $\log N_\ell$  \cite{tsc2014}.

Since our proposed algorithm considers sparse linear representation based on observed entries to build adjacency matrix for each data while TSC-EWZF algorithm builds affinity matrix by thresholding pair-wise distance, TSC-EWZF algorithm could be  a good comparison in respect to the way of building adjacency matrix. 

In our simulations, we consider two cases. The first case illustrates Theorem \ref{thm:1} in this paper,  where all the points are sampled at the same co-ordinates which are the  first $pn$ co-ordinates. 
The second case corresponds to the more general case,  where missing data in each point is randomly sampled at a rounded value of $pn$ co-ordinates thus giving a sampling rate of approximately $p$ and a missing rate of approximately $1-p$. All the results are averaged over $100$ runs for the choice of the data and the sampled elements.


The comparisons for data clustering and  completion are performed using three metrics as explained further. The first metric  is the clustering error.  Clustering error is the ratio calculated by the number of wrongly classified data divided by the number of total data, same as defined in \cite{Yang2015}.

The clustering error for different sampling ratio $p$ in the first case is shown in Fig. \ref{simulation0}, where SSC-EWZF-OO and SSC-EWZF perform the same and show the best clustering accuracy for all sampling ratios from 8\% to 26\%. Furthermore, the plot indicates that the SSC-EWZF-OO and SSC-EWZF cluster data perfectly with $p=0.12$, equivalently at $6$ observations out of $50$ entries, while  SSC-EWZF-OO-LASSO and TSC-EWZF require $8$ and $12$ observations, respectively. Since the rank of each cluster is $3$ and only $6$ observations for each point are needed, we see that observing data at same co-ordinates need much less data for efficient clustering. 
We note that we cannot identify the subspace or complete the data with these observations further illustrating that clustering requires less number of observations than that required for data completion.

The clustering error for random sampling, the scenario described in the Theorem \ref{theorem}, with sampling ratio from 0.30 to 0.95 is shown in Fig \ref{simulation}(a), 
where we note that clustering error with our proposed algorithm, SSC-EWZF-OO,  is the minimum among the four algorithms. Furthermore, the plot shows that the  sampling ratio at which the clustering error hits zero for the algorithms SSC-EWZF-OO,  SSC-EWZF-OO-LASSO and TSC-EWZF are 0.36, 0.38, and 0.46, respectively. Thus, the proposed SSC-EWZF-OO algorithm required least number of observations to efficiently cluster the data. { We further note that the amount of data needed to cluster efficiently for random sampling (36\%) is larger than that for observing data at the same co-ordinates (10\%).  Finally, SSC-EWZF algorithm performs the worst among the four algorithms, where the clustering error does not reach zero even at 95\% observed entries.} 

The second metric is the completion error. Let the recovered matrix using a clustering algorithm be the output of matrix completion using SVT method 
\cite{candes2009exact} on the subspaces found as a result of the subspace clustering and the true matrix be the ground truth of the matrix with missing data. Then the recovery difference is defined as the matrix difference between recovered matrix and true matrix. Thus the completion error is measured by ratio of the Frobenius Norm of the recovery difference to the Frobenius Norm of the true matrix. 
The completion error for different values of $p$ can be seen  in Fig \ref{simulation}(b), 
 where we see completion error is positively correlated with clustering error and small percentage clustering error can result in large percentage completion error. Similar to the clustering error, SSC-EWZF-OO has the lowest completion error among the four algorithms 
and the completion errors for SSC-EWZF-OO,  SSC-EWZF-OO-LASSO and TSC-EWZF becomes zero at sampling ratio of around 0.50, 0.50, and 0.55 respectively, which are larger than the corresponding thresholds for the clustering errors.  Consistent with subspace clustering, SSC-EWZF performs the worst and does not give perfect completion with 95\% observations. 


The third metric is the subspace error.  Since both the completed data and the original data is in union of subspaces, we first find the distance between matched  subspaces in both the completed and original data and average over the different subspaces, where matching that gives minimum error is chosen. The difference in two subspaces is defined in terms of principal angle as follows 
\cite{bjorck1973numerical, wedin1983angles}
\begin{displaymath}
\theta=\arcsin (\|(\M{B}-\M{A}\M{A}^\top\M{B})\|_{\ell_2}),
\end{displaymath}
where $\theta$ is the angle based subspace error, $\M{A}$ is the orthonormal basis of the first subspace, and $\M{B}$ is the orthonormal basis of the second subspace. With the simulation result in Fig \ref{simulation}(c), that subspace error is zero after  46\% sampling ratio for SSC-EWZF-OO, SSC-EWZF-OO-LASSO and TSC-EWZF algorithms, and  60\% for SSC-EWZF algorithm.  For any sampling ratio lower than 46\%, the subspace error for SSC-EWZF-OO is the lowest among the compared algorithms.

We note that the above distance  fails to measure the difference when a low dimensional subspace overlaps with a high dimensional subspace. We thus apply a Grassmann metric based  subspace error \cite{yeschubert} that accounts for the different dimension subspaces to consider the recovery of subspace dimensions. In this metric,  the distance between  two subspaces is given as
\begin{equation}
d({\bf A}, {\bf B}) = (|k-\l|\frac{\pi^2}{4} +\sum_{i=1}^{min(k,l)} \theta_i^2)^{1/2},
\end{equation}
where $k$ is the rank of {\bf A} which is the orthonormal basis of the first subspace, $l$ is the rank of ${\bf B}$ which is the  orthonormal basis of the second subspace,  and $\theta_i$'s are the principal angles \cite{soltanolkotabi2012} between the two subspaces ( $\theta_i$ is calculated by the angle between the $i_{th}$ column vector in ${\bf B}$ and the $i_{th}$ column vector in the basis vector that is obtained by projecting ${\bf A}$ onto ${\bf B}$). 
Grassmann metric based subspace error is shown in Fig \ref{simulation}(d)
and it hits zero at sampling ratio 0.55, 0.55, and 0.6 for SSC-EWZF-OO, SSC-EWZF and TSC respectively. We note that SSC-EWZF does not recover the subspace with the same dimension as the true subspace. Thus, SSC-EWZF recovers either a higher or lower dimensional subspace that overlaps with the true subspace. This  explains  zero orthonormal basis based subspace error for SSC-EWZF while the completion error for SSC-EWZF is not zero. Further,  subspace error at 50\% while there is no error in completion is since the completion error is not exactly zero, but below a threshold which can still cause error in subspace dimensions in some experiments.


\section{Conclusions}
This paper proposes two algorithms for sparse subspace clustering under missing data,  when data is assumed to come from a Union of Subspaces
(UoS) model, using a $\ell_1$ norm minimization based problems. Both the problems use combinations of entry-wise zero-filling and sparse subspace clustering.  Deterministic analysis of sufficient conditions when these algoirthms lead to correct clustering are presented. Extensive set of
simulation results for clustering as well as completion of data
under missing entries, under the UoS model are provided which demonstrate the effectiveness of the algorithm, and demonstrate that accurate
clustering does not imply accurate subspace identification.

We would like to mention that the notion of in-radius is related to the notion of \textbf{permeance} \cite{Lerman2012ArXiv} of data points in a given subspace, quantifying how well the data is distributed inside each subspace. In-radius can be thought of as a worst-case permeance that doesn't scale with the number of data points, while permeance scales with the number of data points and is more of an averaged criteria. Perhaps this is the reason that the primal-dual analysis of SSC under full observation is not able to support the empirical evidence that as the number of points per subspace increases the clustering error goes down dramatically. For subspace clustering such the effect of the number of data points was shown more explicitly in a recent paper \cite{GSC2014}. A connection between these two quantities, namely the in-radius and permeance for subspace clustering under missing data will be undertaken in a future work.


\appendices
 \section{Prior Results}

In this Section, we will present few results that will be used extensively in the proofs. 

{\lemma \label{lemma1} \cite{gritzmann1992inner} Let $r_j({\bf C})$ denote the radius of a largest j-ball contained in {\bf C}, $R_j({\bf C})$ denote the radius of a smallest j-ball containing ${\bf C}$, and ${\bf C}^o$ denote the polar of ${\bf C}$. If the body ${\bf C}$ is a subset of  Minkowski space of dimension $d$ and is symmetric about the origin and $1 \leq j \leq d$, then
\begin{equation}
r_j({\bf C}) R_j({\bf C}^o)=1 \quad \text{and} \quad R_j({\bf C}) r_j({\bf C}^o)=1
\end{equation}
}

{\lemma \label{lemma2} \cite{dim_red_sc_ext} Let ${\bf T}$ be a subset of the column indices of a given matrix ${\bf A}$. All solutions ${\bf c}^*$ of $P({\bf a},{\bf A})$ satisfy ${\bf C}^*_{\bar{{\bf T}}}=0$ if there exists a vector ${\bf c}$ such that ${\bf a}={\bf A}{\bf c}$ with support ${\bf S} \subset {\bf T}$ and a (dual certificate) vector $\bf{v}$ satisfying
\begin{equation} \label{eq: conditions}
{\bf A}_{\bf S}^\top {\bf v} =sign({\bf c}_{\bf S}), \quad \|{\bf A}_{{\bf T}\cap \bar{{\bf S}}}^\top {\bf v}\|_\infty \leq 1, \quad \|{\bf A}_{\bar{T}}^\top {\bf v}\|_\infty <1
\end{equation}
}
\section{Proof of Theorem \ref{theorem}: Deterministic Conditions for SSC-EWZF-OO}\label{apdx:oo}

The full data is located in $n$ dimensional space while under missing data, the algorithm proposed by SSC-EWZF-OO projects data onto low dimension ${\bf I}_{\Omega_i}$ according to represented data ${\bf X}_{\Omega_i}$. 
We apply SVD to analyze the changes of bases (to the projected space) and prove our Theorem \ref{theorem}  by showing that the consitions in the statement of the theorem gives the existence of a solution that satisfies the three dual certificate conditions \eqref{eq: conditions} in Lemma 2. Thus the coefficient vector $\bf{c}_i$ calculated from \eqref{eq: algo} has non-zeros entries only for the data points from the same subspace as ${\bf X}_{\Omega_i}$,  indicating correct clustering. 

Without loss of generality, let $\bf{c}$ denote the notation ${\bf c}_i$ in \eqref{eq: algo}.  According to the optimization problem \eqref{eq: algo} as proposed by SSC-EWZF-OO, with the SVD notation from \eqref{eq: X_omega}, the primal problem  is
\begin{equation}\label{eq: primal}
\begin{split}
& \text{(P)    }\min \|{\bf c}\|_1 \\
&  s.t. \quad {\bf Q}_i^{(\ell)} {\bf \Sigma}_i^{(\ell)}({\bf R}_i^{(\ell)})^\top {\bf a}_i^{(\ell)} \\
& = {\bf I}_{\Omega_i^{(\ell)}}[{\bf Q}_1^{(\ell)} {\bf \Sigma}_1^{(\ell)}({\bf R}_1^{(\ell)})^\top {\bf a}_1^{(\ell)} , ...,  {\bf Q}_{N_l}^{(\ell)} {\bf \Sigma}_{N_l}^{(\ell)}({\bf R}_{N_l}^{(\ell)})^\top {\bf a}_{N_l}^{(\ell)}] {\bf c},
\end{split}
\end{equation}
and the dual problem is
\begin{equation}\label{eq: dual}
\begin{split}
& \text{(D)    }\max \langle {\bf Q}_i^{(\ell)} {\bf \Sigma}_i^{(\ell)}({\bf R}_i^{(\ell)})^\top {\bf a}_i^{(\ell)} , {\bf v} \rangle \\
& s.t.  \quad \forall {j\neq i} \left| \left( {\bf I}_{\Omega_i^{(\ell)}} {\bf Q}_j^{(\ell)} {\bf \Sigma}_j^{(\ell)}({\bf R}_j^{(\ell)})^\top {\bf a}_j^{(\ell)} \right)^\top {\bf v}\right| \leq 1,
\end{split}
\end{equation}

Since ${\bf Q}_i^{(\ell)}$ is a unitary matrix, we have
\begin{equation}
{{\bf Q}_i^{(\ell)}}^\top {\bf Q}_i^{(\ell)} = {{\bf Q}_i^{(\ell)}} {{\bf Q}_i^{(\ell)}}^\top = {\bf I} \in \mathbb{R}^{n\times n}.
\end{equation}

Without changing the problem, we could insert an identity matrix ${{\bf Q}_i^{(\ell)}} {{\bf Q}_i^{(\ell)}}^\top$ into the constraint of the dual problem \eqref{eq: dual}, which gives
\begin{equation} \label{eq: thm_13}
\begin{split}
& \text{(D)    }\max \langle {\bf Q}_i^{(\ell)}{\bf \Sigma}_i^{(\ell)}({\bf R}_i^{(\ell)})^\top {\bf a}_i^{(\ell)} , {\bf v} \rangle \\
& \text{such that } \quad  \forall {j \neq i}\\
&  \| \left( {{\bf Q}_i^{(\ell)}} {{\bf Q}_i^{(\ell)}}^\top {\bf I}_{\Omega_i^{(\ell)}} {\bf Q}_j^{(\ell)} {\bf \Sigma}_j^{(\ell)}({\bf R}_j^{(\ell)})^\top {\bf a}_j^{(\ell)} \right)^\top {\bf v}\|_2 \leq 1,
\end{split}
\end{equation}

Separating ${\bf Q}_i^{(\ell)}$ out of the transpose in the constraint term in \eqref{eq: thm_13} changes the dual problem into
\begin{equation}
\begin{split}
& \text{(D)    }\max \langle {\bf \Sigma}_i^{(\ell)}({\bf R}_i^{(\ell)})^\top {\bf a}_i^{(\ell)} , ({\bf Q}_i^{(\ell)})^\top {\bf v} \rangle \quad \\
& \text{such that } \quad  \forall {j \neq i}\\
&\| \left( ({\bf Q}_i^{(\ell)})^\top {\bf I}_{\Omega_i^{(\ell)}} {\bf Q}_j^{(\ell)} {\bf \Sigma}_j^{(\ell)}({\bf R}_j^{(\ell)})^\top {\bf a}_j^{(\ell)} \right)^\top ({\bf Q}_i^{(\ell)})^\top {\bf v}\|_2 \leq 1,
\end{split}
\end{equation}

Let ${\V \lambda} = ({\bf Q}_i^{(\ell)})^\top {\bf v}$, for simplicity of notations. With the notations $\tilde{{\bf a}}_j^{(\ell)}$, $\tilde{{\bf A}}_{-i}^{(\ell)}$, and $\tilde{{\bf a}}_i^{(\ell)}$ as defined in \eqref{eq: base}, \eqref{eq: base_groups}, and \eqref{eq: notation_a}, the dual problem can be written  concisely as follows
\begin{equation} \label{eq: dual_transform}
\text{(D)    }\max \langle \tilde{{\bf a}}_i^{(\ell)} , {\V \lambda} \rangle \quad s.t. \quad \| (\tilde{{\bf A}}_{-i}^{(\ell)})^\top {\V \lambda} \|_\infty \leq 1.
\end{equation}

Now we will show the existence of a solution that satisfies the three conditions \eqref{eq: conditions} in Lemma 2 by three steps. Within each step, we will show that the selected solution satisfies each of the three conditions respectively. In our proof, ${\bf X}_{\Omega_i}^{(\ell)}$ and ${\bf I}_{\Omega_i^{(\ell)}} {\bf X}_{-i, \Omega}$ correspond to ${\bf a}$ and $\bf {A}$ in Lemma \ref{lemma2} respectively. 

{\bf Step 1 }

Let ${\bf S}$ be the support of the solution ${\bf c}$ in \eqref{eq: algo}, ${\bf v}_i^{(\ell)}$ be a solution to the dual problem \eqref{eq: dual}, and ${\bf X}_{{\bf S}, \Omega}$ be the data corresponding to ${\bf S}$ from ${\bf X}_{-i, \Omega}$. Then the objective function value of the primal problem \eqref{eq: primal} is
\begin{equation}
\|{\bf c}\|_1 =\|{\bf c_S}\|_1 =\langle {\bf c_S}, sign({\bf c_S}) \rangle.
\end{equation}
The objective function value of the dual problem \eqref{eq: dual} is
\begin{equation}
\langle {\bf X}_{\Omega_i}^{(\ell)}, {\bf v}_i^{(\ell)}\rangle
= \langle {\bf I}_{\Omega_i^{(\ell)}} {\bf X}_{{\bf S}, \Omega} {\bf c_S}, {\bf v}_i^{(\ell)}\rangle
=  \langle  {\bf c_S}, ({\bf I}_{\Omega_i^{(\ell)}} {\bf X}_{{\bf S}, \Omega})^\top {\bf v}_i^{(\ell)}\rangle,
\end{equation}

For linear programming problem, strong duality always holds \cite{boyd2004convex}, and thus $\langle {\bf c_S}, sign({\bf c_S}) \rangle = \langle  {\bf c_S}, ({\bf I}_{\Omega_i^{(\ell)}} {\bf X}_{{\bf S}, \Omega})^\top {\bf v}_i^{(\ell)}\rangle$. Since $sign(\bf{c_S})$ is the unique maximizer of $\max _{{\bf a}: \|{\bf a}\|_\infty \leq 1} \langle {\bf c_S}, {\bf a}\rangle $, we have
\begin{equation}
({\bf I}_{\Omega_i^{(\ell)}} {\bf X}_{{\bf S}, \Omega})^\top {\bf v}_i^{(\ell)}=sign({\bf c_S}).
\end{equation}
Thus, the solution satisfies the first condition of Lemma \ref{lemma2}.\\

{\bf Step 2 }

Since ${\bf v}_i^{(\ell)}$ is a solution to the dual problem, ${\bf v}_i^{(\ell)}$ has to satisfy the constraint of the dual problem \eqref{eq: dual}, and thus
\begin{equation}
\| ({\bf I}_{\Omega_i^{(\ell)}} {\bf X}_{-i, \Omega})^\top {\bf v}_i^{(\ell)} \|_\infty \leq 1.
\end{equation}
Thus, the solution satisfies the second condition of Lemma \ref{lemma2}.\\

{\bf Step 3 }

The constraint $\quad \| (\tilde{{\bf A}}_{-i}^{(\ell)})^\top {\V \lambda}_i^{(\ell)} \|_\infty \leq 1$ given by transformed dual problem in \eqref{eq: dual_transform} indicates
\begin{equation}\label{eq: 31}
\|{\V \lambda}_i^{(\ell)}\|_2 \leq R(\mc{P}^o (\tilde{\bf{A}}_{-i}^{(\ell)})),
\end{equation}
where $R(\mc{P}^o (\tilde{{\bf A}}_{-i}^{(\ell)}))$ is the radius of the smallest ball that contains $\mc{P}^o(\tilde{\bf{A}}_{-i}^{(\ell)})$.

Lemma \ref{lemma1} implies that $R(\mc{P}^o (\tilde{\bf{A}}_{-i}^{(\ell)})) = \frac{1}{r(\mc{P}(\tilde{{\bf A}}_{-i}^{(\ell)}))}$, where $r(\mc{P}(\tilde{{\bf A}}_{-i}^{(\ell)}))$ is the radius of the largest ball contained in $\tilde{{\bf A}}_{-i}^{(\ell)}$, thus from \eqref{eq: 31} we have
\begin{equation}\label{eq: fact}
\frac{1}{\|{\V \lambda}_i^{(\ell)}\|_2} \geq r(\mc{P}(\tilde{{\bf A}}_{-i}^{(\ell)})).
\end{equation}
By the assumption in the theorem  \eqref{eq:theorem_eq}, we have 
\begin{equation}\label{eq: fact1}
\left| 
\frac{({\V \lambda}_i^{(\ell)})^\top}{\|({\V \lambda}_i^{(\ell)})^\top\|_2} ({\bf Q}_i^{(\ell)})^\top {\bf I}_{\Omega_i^{(\ell)}} {\bf I}_{\Omega_j^{(k)}} {\bf U}^{(k)} {\bf a}_j^{(k)}
\right|
< r(P(\tilde{{\bf A}}_{-i}^{(\ell)})).
\end{equation}
From \eqref{eq: fact} and \eqref{eq: fact1}, we get
\begin{equation}
\left| 
\frac{({\V \lambda}_i^{(\ell)})^\top}{\|({\V \lambda}_i^{(\ell)})^\top\|_2} ({\bf Q}_i^{(\ell)})^\top {\bf I}_{\Omega_i^{(\ell)}} {\bf I}_{\Omega_j^{(k)}} {\bf U}^{(k)} {\bf a}_j^{(k)}
\right|
< \frac{1}{\|{\V \lambda}_i^{(\ell)}\|_2}.
\end{equation}
This is equivalent to
\begin{equation}\label{eq: dual_con3}
\left| 
({\V \lambda}_i^{(\ell)})^\top ({\bf Q}_i^{(\ell)})^\top 
{\bf I}_{\Omega_i^{(\ell)}} {\bf I}_{\Omega_j^{(k)}} {\bf U}^{(k)} {\bf a}_j^{(k)}
\right|
< 1.
\end{equation}

Since ${\V \lambda}_i^{(\ell)} = ({\bf Q}_i^{(\ell)})^\top {\bf v}_i^{(\ell)}$ and ${\bf I}_{\Omega_i^{(\ell)}} {\bf I}_{\Omega_j^{(k)}} {\bf U}^{(k)} {\bf a}_j^{(k)} ={\bf I}_{\Omega_i^{(\ell)}} {\bf X}_{\Omega_j}^{(k)}$, \eqref{eq: dual_con3} reduces to
\begin{equation}
\left| \langle {\bf v}_i^{(\ell)}, {\bf I}_{\Omega_i^{(\ell)}} {\bf X}_{\Omega_j}^{(k)} \rangle \right|< 1 \quad \forall _{k\neq l, j},
\end{equation}
thus showing that the solution satisfies the third condition in Lemma \ref{lemma2}.
\section{Proof of Theorem \ref{corollary3}: Deterministic Conditions for SSC-EWZF}\label{apdx:res}

The proof follows on the same lines as Theorem \ref{theorem}, by changing the $\tilde{\bf a}_j^{(k)} = ({\bf Q}_i^{(l)})^\top {\bf I}_{\Omega_i^{(l)}} {\bf V}_{\Omega_j}^{(k)} {\bf a}_j^{(k)}$ in SSC-EWZF-OO to $\tilde{\bf a}_j^{(k)} = ({\bf Q}_i^{(l)})^\top {\bf V}_{\Omega_j}^{(k)} {\bf a}_j^{(k)}$ in SSC-EWZF. 
Details are as follows. 

The primal problem in SSC-EWZF algorithm to be solved is
\begin{equation}\label{eq: P_res}
\begin{split}
& \text{(P)    }\min \|{\bf c}\|_1 \\
&  s.t. \quad {\bf Q}_i^{(\ell)} {\bf \Sigma}_i^{(\ell)}({\bf R}_i^{(\ell)})^\top {\bf a}_i^{(\ell)} \\
& = [{\bf Q}_1^{(\ell)} {\bf \Sigma}_1^{(\ell)}({\bf R}_1^{(\ell)})^\top {\bf a}_1^{(\ell)} , ...,  {\bf Q}_{N_l}^{(\ell)} {\bf \Sigma}_{N_l}^{(\ell)}({\bf R}_{N_l}^{(\ell)})^\top {\bf a}_{N_l}^{(\ell)}] {\bf c},
\end{split}
\end{equation}
and the corresponding dual problem is
\begin{equation}\label{eq: dual-38}
\begin{split}
& \text{(D)    }\max \langle {\bf Q}_i^{(\ell)} {\bf \Sigma}_i^{(\ell)}({\bf R}_i^{(\ell)})^\top {\bf a}_i^{(\ell)} , {\bf v} \rangle \\
& s.t.  \quad \forall {j\neq i} \left| \left( {\bf Q}_j^{(\ell)} {\bf \Sigma}_j^{(\ell)}({\bf R}_j^{(\ell)})^\top {\bf a}_j^{(\ell)} \right)^\top {\bf v}\right| \leq 1.
\end{split}
\end{equation}
An identity matrix ${\bf I} = {{\bf Q}_i^{(\ell)}}^\top {\bf Q}_i^{(\ell)}$ is inserted in the constraint in \eqref{eq: dual-38} to obtain
\begin{equation}\label{eq: dual-39} 
\begin{split}
& \text{(D)    }\max\langle {\bf Q}_i^{(\ell)}{\bf \Sigma}_i^{(\ell)}({\bf R}_i^{(\ell)})^\top {\bf a}_i^{(\ell)} , {\bf v} \rangle \\
& s.t. \quad  \forall {j \neq i} \| \left( {{\bf Q}_i^{(\ell)}} {{\bf Q}_i^{(\ell)}}^\top {\bf Q}_j^{(\ell)} {\bf \Sigma}_j^{(\ell)}({\bf R}_j^{(\ell)})^\top {\bf a}_j^{(\ell)} \right)^\top {\bf v}\|_2 \leq 1.
\end{split}
\end{equation}
Let $\V{\lambda}  = {{\bf Q}_i^{(\ell)} }^\top  {\bf v}$, then the dual problem \eqref{eq: dual-39} becomes
\begin{equation}\label{eq: dual-40} 
\begin{split}
& \text{(D) } \max  \langle {\bf \Sigma}_i^{(\ell)}({\bf R}_i^{(\ell)})^\top {\bf a}_i^{(\ell)} , \V{\lambda} \rangle \\
& s.t. \quad  \forall {j \neq i} \| \left( {{\bf Q}_i^{(\ell)}}^\top {\bf Q}_j^{(\ell)} {\bf \Sigma}_j^{(\ell)}({\bf R}_j^{(\ell)})^\top {\bf a}_j^{(\ell)} \right)^\top \V{\lambda}\|_2 \leq 1.
\end{split}
\end{equation}
Let  $\tilde{\bf b}_j^{(k)} = ({\bf Q}_i^{(\ell)})^\top {\bf V}_{\Omega_j}^{(k)} {\bf a}_j^{(k)}$. Note that $\tilde{\bf b}_i^{(\ell)} = {\bf \Sigma}_i^{(\ell)}({\bf R}_i^{(\ell)})^\top {\bf a}_i^{(\ell)} $. Further let $\tilde{\bf B}_{-i}^{(\ell)}$ be the $n\times (N_\ell-1)$ matrix with columns as $\tilde{{\bf b}}_j^{(\ell)}, j\neq i$, then \eqref{eq: dual-40} can be written as
\begin{equation}\label{eq: dual-41} 
 \text{(D) } \max \langle \tilde{\bf b}_i^{(\ell)} , \V{\lambda} \rangle
 s.t. \quad  \| \tilde{\bf B}_{-i}^{(\ell)} \V{\lambda}\|_\infty \leq 1.
\end{equation}

We now prove the theorem by showing that there exist a solution which satisfies the three conditions \eqref{eq: conditions} in Lemma 2 as in Appendix \ref{apdx:oo}. In this proof, $X_{\Omega_i}^{(\ell)}$ and ${\bf X}_{-i, \Omega}$ correspond to ${\bf a}$ and ${\bf A}$ respectively in Lemma \ref{lemma2}. 

{\bf Step 1}

Let ${\bf S}$ be the support of the solution ${\bf c}$, ${\bf v}_i^{(\ell)}$ be a solution to the dual problem \eqref{eq: dual-41}, and ${\bf X}_{{\bf S},\Omega}$ be the data corresponding to ${\bf S}$ from ${\bf X}_{-i, \Omega}$. The objective function value \eqref{eq: P_res} can be written as 
\begin{equation}
\|{\bf c}\|_1 =\|{\bf c}_{\bf S}\|_1 =\langle{\bf c}_{\bf S}, sign({\bf c}_{\bf S}) \rangle .
\end{equation}
The objective function value of the dual problem is 
\begin{equation} \label{eq: tp43}
\langle \tilde{{\bf b}}_i^{(\ell)},   {\bf Q}_i^{(\ell)} {\bf v}_i^{(\ell)}\rangle
=\langle {\bf X}_{\Omega_i^{(\ell)}},   {\bf v}_i^{(\ell)}  \rangle.
\end{equation}
From the constraints in \eqref{eq: P_res} and \eqref{eq: tp43}, we have
\begin{equation} \label{eq: tp44}
\langle  {\bf X}_{\Omega_i^{(\ell)}},   {\bf v}_i^{(\ell)}  \rangle =
\langle {\bf X}_{{\bf S},\Omega} {\bf c}_{\bf S}, {\bf v}_i^{(\ell)} \rangle =
\langle  {\bf c}_{\bf S}, ({\bf X}_{{\bf S},\Omega})^\top {\bf v}_i^{(\ell)} \rangle .
\end{equation}
For linear programming problem in \eqref{eq: P_res}, strong duality  holds, thus giving
\begin{equation}
\langle {\bf c}_{\bf S}, sign({\bf c}_{\bf S}) \rangle  = \langle  {\bf c}_{\bf S}, ({\bf X}_{{\bf S},\Omega})^\top {\bf v}_i^{(\ell)} \rangle .
\end{equation}
As $sign({\bf C_S})$ is the unique optimizer, thus
\begin{equation}
({\bf X}_{{\bf S},\Omega})^\top {\bf v}_i^{(\ell)}  =sign({\bf c}_{\bf S}).
\end{equation}
Thus, the solution satisfies the first condition in Lemma \ref{lemma2}.\\

{\bf Step 2}

Since ${\bf v}_i^{(\ell)}$ is a solution to the dual problem \eqref{eq: dual-38}, the constraint in  \eqref{eq: dual-38} has to be satisfied, and thus
\begin{equation}
\forall _{j\neq i} \left| ({\bf X}_j^{(\ell)}) ^\top {\bf v}_i^{(\ell)} \right| \leq 1.
\end{equation}
This is equivalent  to
\begin{equation}
\|{\bf X}_{-i, \Omega}^{(\ell)} {\bf v}_i^{(\ell)}\|_\infty \leq 1.
\end{equation}
Thus, the solution satisfies the second condition in Lemma \ref{lemma2}

{\bf Step 3}

The constraint given by \eqref{eq: dual-41} implies that the solution $\V{\lambda}_i^{(\ell)}$ satisfies
\begin{equation}\label{eq: tp3_1}
\|\V{\lambda}_i^{(\ell)}\|_2 \leq R(\mc{P}^o(\tilde{\bf B}_{-i}^{(\ell)})).
\end{equation}

From Lemmas \ref{lemma1} and \eqref{eq: tp3_1}, we obtain
\begin{equation}\label{eq: tp3_2}
\frac{1}{\|\V{\lambda}_i^{(\ell)}\|_2}  \geq \frac {1}{R(\mc{P}^o(\tilde{\bf B}_{-i}^{(\ell)}))} = r(\mc{P}(\tilde{\bf B}_{-i}^{(\ell)})).
\end{equation}

Based on the assumption in the statement of  Theorem \ref{corollary3}, we have
\begin{equation}\label{eq: tp3_3}
\left| 
\frac{({\V \lambda}_i^{(\ell)})^\top}{\|({\V \lambda}_i^{(\ell)})^\top\|_2} ({\bf Q}_i^{(\ell)})^\top  {\bf V}_{\Omega_j}^{(k)} {\bf a}_j^{(k)}
\right|
< 
r(\mc{P}(\tilde{{\bf B}}_{-i}^{(\ell)})).
\end{equation}

From \eqref{eq: tp3_2} and \eqref{eq: tp3_3}, we obtain
\begin{equation}\label{eq: tp3_4}
\left| 
{({\V \lambda}_i^{(\ell)})^\top} ({\bf Q}_i^{(\ell)})^\top  {\bf V}_{\Omega_j}^{(k)} {\bf a}_j^{(k)}
\right|
< 
1.
\end{equation}
This is equivalent to
\begin{equation}\label{eq: tp3_4}
\forall {k\neq \ell, j}
\left|
\langle
{\bf v}_i^{(\ell)} , {\bf X}_{\Omega_j}^{(k)} 
\rangle
\right|
< 
1,
\end{equation}
thus showing that the solution satisfies the third condition in Lemma \ref{lemma2}.

\section{Proof of Theorem \ref{thm:1}: Same Location Sampling Case} \label{apdx:same2}
We prove theorem  \ref{thm:1} by showing that the assumptions in the statement of  Theorem \ref{thm:1} indicates the existence of a solution that satisfies the three conditions in Lemma \ref{lemma2},
similar to that  in Appendix \ref{apdx:oo}

As the observation points are the same for each data, SSC-EWZF solves the following optimization problem
\begin{equation} \label{eq: primal1}
\text{(P)    }\min \|{\bf c}\|_1 \quad \text{s.t.} \quad {\bf I}_{\Omega} {\bf X}_i^{(\ell)} = {\bf I}_{\Omega} {\bf X}_{-i}^{(\ell)}{\bf c}.
\end{equation}

Without loss of generality, let ${\bf Y}_i^{(\ell)} = {\bf I}_{\Omega} {\bf X}_i^{(\ell)}$ and ${\bf Y}_{-i}^{(\ell)} = {\bf I}_{\Omega}{\bf X}_{-i}^{(\ell)}$. Then the primal problem in \eqref{eq: primal1} reduces to
\begin{equation}\label{eq: primal2}
\text{(P)    }\min \|{\bf c}\|_1 \quad \text{s.t.} \quad {\bf Y}_i^{(\ell)} = {\bf Y}_{-i}^{(\ell)}{\bf c}.
\end{equation}

The dual problem corresponding to \eqref{eq: primal2} is
\begin{equation}\label{eq: dual_v2}
\text{(D)    }\max \langle {\bf Y}_i^{(\ell)} , {\bf v}\rangle \quad \text{s.t.} \quad \|({\bf Y}_{-i}^{(\ell)})^\top {\bf v}\|_\infty \leq 1.
\end{equation}

Let ${\bf V}_\Omega^{(\ell)} ={\bf I}_\Omega {\bf U}^{(\ell)}$ denote the modified basis after considering the partial observation, which gives ${\bf Y}_i^{(\ell)} ={\bf V}_\Omega^{(\ell)}{\bf a}_i^{(\ell)}$ and ${\bf Y}_{-i}^{(\ell)} = {\bf V}_\Omega^{(\ell)}{\bf A}_{-i}^{(\ell)}$.  Now we will show the existence of a solution that satisfies the three conditions in Lemma \ref{lemma2} by three steps.  Note that ${\bf Y}_i^{(\ell)}$ and ${\bf Y}_{-i}^{(\ell)}$ are corresponding to ${\bf a}$ and ${\bf A}$ in Lemma \ref{lemma2} respectively. 

{\bf Step 1:}
Let S be the support of the solution, the primal problem  gives 
\begin{equation}\label{eq: v2_1}
\|{\bf c}\|_1  = \langle  {\bf c}_S, sign({\bf c}_S)\rangle,
\end{equation}
and  the dual problem gives
\begin{equation}\label{eq: v2_2}
\langle {\bf Y}_i^{(\ell)}, {\bf v}_i^{(\ell)}\rangle = \langle {\bf Y}_{-i, {\bf S}}^{(\ell)}{\bf c}_S, {\bf v}_i^{(\ell)}\rangle = \langle {\bf c}_S, ({\bf Y}_{-i, {\bf S}}^{(\ell)})^\top {\bf v}_i^{(\ell)}\rangle,
\end{equation}
 where ${\bf Y}_{-i, {\bf S}}^{(\ell)}$ is the support in ${\bf Y}_{-i}^{(\ell)}$. Since the strong duality holds for primal problem which is a linear programming problem, we have
\begin{equation}\label{eq: v2_3}
\|{\bf c}\|_1 =\langle {\bf Y}_i^{(\ell)}, {\bf v}_i^{(\ell)}\rangle.
\end{equation}

Since $sign({\bf c_S})$ is the unique optimizer, from \eqref{eq: v2_1}, \eqref{eq: v2_2} and \eqref{eq: v2_3} we have 
\begin{equation}
({\bf Y}_{-i,{\bf S}}^{(\ell)})^\top {\bf v}_i^{(\ell)} = sign({\bf c}_S).
\end{equation}
Thus, the solution satisfies the first condition in Lemma \ref{lemma2}.

{\bf Step 2:} Since the solution of the dual problem has to satisfy the constraint in the dual problem \eqref{eq: dual_v2}, we have
\begin{equation}
\|({\bf Y}_{-i}^{(\ell)})^\top {\bf v}_i^{(\ell)} \|_\infty \leq 1.
\end{equation}
Thus, the solution satisfies the second condition in Lemma \ref{lemma2}.

{\bf Step 3: }
From the dual problem \eqref{eq: dual_v2}, we have
\begin{equation}
\begin{split}
&OptSolD({\bf Y}_i^{(\ell)}, {\bf Y}_{-i}^{(\ell)}) \\
& = \arg \max _{\bf v} \langle {\bf V}_\Omega^{(\ell)} {\bf a}_i^{(\ell)}, {\bf v}\rangle \quad \text{s.t.}\quad \|({\bf V}_\Omega^{(\ell)}{\bf A}_{-i}^{(\ell)})^\top {\bf v}\|_\infty \leq 1 \\
& =  \arg \max _{\bf v} \langle  {\bf a}_i^{(\ell)}, ({\bf V}_\Omega^{(\ell)})^\top {\bf v} \rangle \quad \text{s.t.}\quad \|({\bf A}_{-i}^{(\ell)})^\top ({\bf V})_\Omega^{(\ell)})^\top {\bf v}\|_\infty \leq 1 .
\end{split}
\end{equation}

Let ${\bf {\bar{\V{\lambda}}}} = {{\bf V}_\Omega^{(\ell)}}^\top {\bf v}$ and ${\bar{\V{\lambda}}}_i^{(\ell)} \in OptSolD({\bf a}_i^{(\ell)}, {\bf A}_{-i}^{(\ell)})$, we see that  
\begin{equation}\label{eq: v2_9}
{\bf v}_i^{(\ell)} = ({{\bf V}_\Omega^{(\ell)}}^\top )^\dagger {\bf {\bar{\V{\lambda}}}}_i^{(\ell)},
\end{equation}
 is  a solution to the dual problem, where
$({{\bf V}_\Omega^{(\ell)}}^\top )^\dagger = {{\bf V}_\Omega^{(\ell)}} ({{\bf V}_\Omega^{(\ell)}}^\top {{\bf V}_\Omega^{(\ell)}})^{-1}$.

Since ${\bar{\V{\lambda}}}_i^{(\ell)} \in OptSolD({\bf a}_i^{(\ell)}, {\bf A}_{-i}^{(\ell)})$, with the result from Lemma \ref{lemma1}, we get
\begin{equation}\label{eq: v2-10}
\bar{\V{\lambda}}_i^{(\ell)} \leq R(\mc{P}^o({\bf A}_{-i}^{(\ell)})) =\frac{1}{r(\mc{P}({\bf A}_{-i}^{(\ell)}))}.
\end{equation}

From the assumptions  in the statement of Theorem \ref{thm:1}, we have
\begin{equation} \label{eq: v2_11}
\left| { ( \bar{ \V {\lambda}}_i^{(\ell)} )^\top }   ({\bf V}_\Omega^{(\ell)})^\dagger  {\bf V}_\Omega^{(k)} {\bf a}_j^{(k)} \right| 
< 
\| \bar{\V{\lambda}}_i^{(\ell)} \|_2 \times r(\mc{P} ({\bf A}_{-i}^{(\ell)}) ).
\end{equation}

From \eqref{eq: v2_9}, \eqref{eq: v2-10}, and \eqref{eq: v2_11}, we have
\begin{equation}
\left|
\langle {\bf v}_i^{(\ell)}, {\bf X}_j^{(k)}
\right|<1, \forall {\ell \neq k},
\end{equation}
thus showing that the solution satisfies the third condition in Lemma \ref{lemma2}.



\bibliographystyle{IEEEbib}
\bibliography{SSmC_Geom,SSmCbibliography,NewBib}

\end{document}